\RequirePackage{amsmath}
\documentclass[runningheads,envcountsame,envcountsect]{llncs}

\usepackage[T1]{fontenc}
\usepackage{amsfonts,graphicx,hyperref,color,listings,enumerate,todonotes}
\usepackage[capitalise]{cleveref}

\newcommand*{\nn}{\mathbb{N}}

\newcommand*{\from}{\colon} 
  
\newcommand*{\emptyw}{\varepsilon}
\newcommand*{\card}{\#}

\newcommand*{\A}{\mathcal{A}}
\newcommand*{\B}{\mathcal{B}}
\newcommand*{\lan}{\mathcal{L}}

\newcommand{\alphabet}{{\ensuremath{\mathrm{alph}}}}

\newcommand*{\weak}{\mathsf{w}}
\newcommand*{\unbound}{\mathsf{u}}
\newcommand*{\strong}{\mathsf{s}}

\newcommand*{\repets}{\mathrm{\Omega}}

\newcommand*{\maxlen}[1]{\lceil#1\rceil}
\newcommand*{\minlen}[1]{\lfloor#1\rfloor}

\begin{document}

\renewcommand{\thelstlisting}{\arabic{lstlisting}}
 
\title{Circularity and repetitiveness in non-injective DF0L systems\thanks{The first author was supported by the CTU Global Postdoc Fellowship program.}}

\titlerunning{Circularity and repetitiveness}
 
\author{Herman Goulet-Ouellet \and Karel Klouda \and Štěpán Starosta}

\institute{Czech Technical University in Prague, Faculty of Information Technology
	\email{\{herman.goulet.ouellet,karel.klouda,stepan.starosta\}@fit.cvut.cz}}

\maketitle              

\begin{abstract}
	We study circularity in DF0L systems, a generalization of D0L systems. We focus on two different types of circularity, called weak and strong circularity. When the morphism is injective on the language of the system, the two notions are equivalent, but they may differ otherwise. Our main result shows that failure of weak circularity implies unbounded repetitiveness, and that unbounded repetitiveness implies failure of strong circularity. This extends previous work by the second and third authors for injective systems. To help motivate this work, we also give examples of non-injective but strongly circular systems.
\end{abstract}

\section{Introduction}

The notion of circularity in D0L systems appeared in several different forms, starting with the work of Cassaigne on pattern avoidance~\cite{Ca94} and Mignosi and Séébold on repetitiveness~\cite{MiSe}. It is closely related with recognizability, a key notion from symbolic dynamics~\cite{th/Kyriakoglou2019}. The study of circularity in D0L systems is also motivated by its connection with bispecial factors~\cite{Ca,Frid1999,Klouda2019}, in turn linked with factor complexity, return words, Rauzy graphs, etc. One of the main motivations behind this paper is to better understand under which conditions weak and strong circularity can be safely decided, with effective calculation of the circularity thresholds being the ultimate goal.
 
In this paper we work with DF0L systems, a generalization of D0L systems which allows for multiple axioms.
The need for this added generality arises naturally when working within a D0L system, for example when one wants to take powers of the morphism without changing the language (see \Cref{e:power-axioms}).

We moreover study systems where the morphism is not necessarily injective.
This is a departure from many earlier papers, where injectivity is needed for most results.
We focus on two notions of circularity, called weak and strong circularity.
While these are equivalent for injective systems (and more generally \emph{eventually injective} systems, defined in \Cref{s:df0l}), they might differ in general.
 
Our main result, \Cref{t:circ-repet}, extends to the non-injective case a characterization from~\cite{Klouda2019} stating that strong circularity is equivalent to unbounded repetitiveness.
In the non-injective case, the result is no longer a characterization.
A consequence of our main result is that when the system is not unboundedly repetitive (a decidable condition by \cite{KlSt13}), then we can safely calculate the weak circularity threshold.
Since practical calculation of circularity thresholds is also one of our goals, we describe the relation between weak and strong circularity thresholds in \Cref{p:weak-vs-strong}, and provide pseudocode in an appendix.
Finally, we also show that in a strongly circular system, weak circularity is preserved when taking powers of the morphism, something which is not true in general.
The questions of whether strong circularity is decidable or preserved under taking powers of the morphism remain open.

\section{DF0L systems}
\label{s:df0l}

Let $\A^*$ be the set of all finite words over a finite alphabet $\A$. The empty word is denoted by $\emptyw$ and we let $\A^+ = \A^*\setminus\{\emptyw\}$.

\begin{definition}
	A \emph{DF0L system} is a triplet $S = (\A, \varphi, W)$ where $\varphi$ is a morphism $\A^*\to\A^*$ and $W \subseteq\A^+$ is a non-empty finite set of non-empty words called \emph{axioms}.
	If $\varphi$ is non-erasing, i.e.\ $\varphi(a)\neq\emptyw$ for all $a\in A$, then $S$ is called a \emph{PDF0L system}.
\end{definition}

We refer to \cite{RoSa80} for more details on the terminology of L systems.
When $W$ consists of a single element, $S$ is also called a D0L system, and a PD0L system if $\varphi$ is non-erasing.
We note that DF0L systems can arise naturally even when working in the setting of D0L systems, as shown in \Cref{e:power-axioms}.

We define the language of a DF0L system $S = (\A,\varphi,W)$ by
\[
	\lan(S) = \{ u\in A^* \mid \exists k\geq 0, \exists w\in W, \varphi^k(w)\in A^*uA^* \}.
\]
Note that membership in $\lan(S)$ is decidable by \cite[Lemma 3]{Salo2017}.

\begin{example}\label{ex:thue-morse}
	The \emph{Thue--Morse system} is the D0L system $S = (\{a,b\},\varphi,\{a\})$ where $\varphi\from a\mapsto ab, b\mapsto ba$. In this system
	\begin{equation*}
		\lan(S) = \{\emptyw,a, b, aa, ab, ba, bb, aab, aba, abb, baa, bab, bba, aaba, aabb,\dots\}.
	\end{equation*}
\end{example}

Let us say that a DF0L system $(\A,\varphi,W)$ is \emph{injective} if $\varphi$ is injective on $\lan(S)$. Many well-known examples of D0L systems are injective, like the Thue--Morse system above. We next introduce a weaker form of injectivity, under which the equivalence between the two kinds of circularity still holds (\Cref{p:weak-vs-strong}).

\begin{definition}
	A DF0L system $S = (\A,\varphi,W)$ is called \emph{eventually injective} if the following set is finite:
	\begin{equation*}
        \mathrm{\Delta}_S = \left \{ \{u,v\}\subseteq \lan(S) \mid u\neq v, \varphi(u)=\varphi(v) \right \}
	\end{equation*}
\end{definition}

Failure of injectivity of a PDF0L system $S = (\A,\varphi,W)$ can be measured by
\begin{equation*}
	\delta_S =  \max\{|\varphi(u)| : \exists v, \{u,v\}\in \mathrm\Delta_S\},
\end{equation*}
with $\delta_S=0$ when $\mathrm\Delta_S=\emptyset$. Eventual injectivity is equivalent to $\delta_S<\infty$. We do not know whether eventual injectivity, or even injectivity, is decidable. We next give two examples non-injective systems, one of which is eventually injective.

\begin{example}\label{ex:eventually-injective}
	Let $S = (\{a,b,c\}, \varphi, \{a\})$ where
	$ \varphi\from a\mapsto abacc, b\mapsto aba, c\mapsto aba$. This system is eventually injective, with
	\begin{equation*}
		\mathrm\Delta_S = \{ \{b,c\}, \{ba,ca\}, \{ab,ac\}, \{bac,cab\} \},
	\end{equation*}
	and thus $\delta_S = 11$.
\end{example}

\begin{example}\label{ex:non-eventually-injective}
	Let $S = (\{a,b,c\}, \varphi, \{a\})$ where
	$ \varphi\from a\mapsto abaca, b\mapsto aba, c\mapsto aba$.
    Unlike the previous example, this system is not eventually injective. We can obtain an infinite sequence of elements $\{u_n,v_n\}_{n\in\nn}$ in $\mathrm\Delta_S$ as follows: let
	\begin{equation*}
		u_1 = aca,\ u_{n+1} = u_1\varphi(u_n),\quad v_1 = aba,\ v_{n+1} = v_1\varphi(v_n).
	\end{equation*}
    It is clear that $u_n\neq v_n$ and $\varphi(u_n) = \varphi(v_n)$. Observing that $bu_n,av_n\in\lan(S)$, one can show by induction on $n$ that $u_n,v_n\in\lan(S)$ for all $n\in\nn$.
\end{example}

A common technique in the study of D0L systems is to pass to a power of the morphism, for instance to ensure certain growth conditions (as in the proof of \Cref{t:circ-repet}), while preserving the language of the system. Given a system $S = (\A,\varphi,W)$, let us define the $k$-th power of $S$ ($k\geq 1$) as the system
\[
	S^k = (\A,\varphi^k,\{\varphi^i(w) \mid w\in W, 0\leq i<k\}),
\]
which satisfies $\lan(S^k) = \lan(S)$. The next example shows the necessity of using multiple axioms when defining $S^k$.

\begin{example}\label{e:power-axioms}
	Let $S = (\A,\varphi,b)$ where $\varphi\from a\mapsto cb, b\mapsto ad, c\mapsto c, d\mapsto d$.
	Let $S' = (\A,\varphi^2,b)$. The language $\lan(S')$ is contained in $\lan(S)$, but the inclusion is strict, since for instance $ad\notin\lan(S')$.
\end{example}

\section{Circularity}

In this section, we recall two notions of circularity for DF0L systems called \emph{weak} and \emph{strong} circularity. We follow in part the exposition from~\cite{Klouda2019}. We also sketch algorithms for computing the weak and strong circularity thresholds. We start with the definition of interpretations.

\begin{definition}
	Let $S = (\A,\varphi,W)$ be a DF0L system. An \emph{interpretation} of $u$ in $S$ is a triplet $(s,w,t)$ such that $\varphi(w)=sut$ and $w\in\lan(S)$. An interpretation is \emph{minimal} if $|s|<|\varphi(a)|$ and $|t|<|\varphi(b)|$, where $a$ and $b$ are respectively the first and last letters of $w$.
\end{definition}

Given a morphism $\varphi$, we let
\[
	\minlen\varphi = \min(|\varphi(a)| : a\in\A) \quad \text{ and } \quad \maxlen\varphi = \max(|\varphi(a)| : a\in\A).
\]

\begin{lemma}\label{l:inter-length}
	Let $S=(\A,\varphi,W)$ be a PDF0L system and $u\in\lan(S)$. If $(s,w,t)$ is a minimal interpretation of $u$ in $S$, then
	\[
		|u|/\maxlen\varphi\leq |w|\leq 2+(|u|-2)/\minlen\varphi.
	\]
\end{lemma}

\begin{proof}
	The leftmost inequality follows from $|u|\leq|\varphi(w)|\leq \maxlen\varphi\cdot|w|$. For the rightmost inequality, first observe that if $|w|=1$ or $|w|=0$, then the statement becomes trivial. Thus we assume that $|w|\geq 2$. Let $a$ and $b$ be respectively the first and last letters of $w$, and let $w=aw'b$. Since the interpretation is minimal, there exist $s',t'\in \A^+$ and $s'',t''\in \A^*$ such that $\varphi(a) = s''s'$, $\varphi(b)=t't''$ and $u=s'\varphi(w')t'$. It follows that
	\[
		|u| = |s'|+|t'|+|\varphi(w')|
		\geq 2 + |w'|\minlen\varphi
		= 2+(|w|-2)\minlen\varphi.
	\]
	Since $\varphi$ is non-erasing we have $\minlen\varphi > 0$, and we obtain the desired inequality by solving for $|w|$.
	\qed
\end{proof}

Thus in a PDF0L system, the set of minimal interpretations of a given word is finite and computable. We may write a simple algorithm which computes all minimal interpretations of the words of length $n$ by looking at the images of all words of length $\lfloor 2+(n-2)/\minlen\varphi\rfloor$. A pseudocode implementation is provided in \cref{s:pseudocode}, Algorithm \ref{lst:inter}.

\begin{example}\label{ex:thue-morse-inter}
	In the Thue--Morse system (\cref{ex:thue-morse}), the word $aba$ admits exactly two minimal interpretations, namely $(b,bb,\emptyw)$ and $(\emptyw,aa,b)$. On the other hand the word $aa$ admits only one minimal interpretation, namely $(b,ba,a)$.
\end{example}

Next we turn to a finer analysis of how interpretations behave in a DF0L system $S = (\A,\varphi,W)$. A pair $(u',u'')$ is said to be \emph{compatible} with an interpretation $(s,w,t)$ of $u'u''$ in $S$ if there is a pair $(w',w'')$ such that $w'w''=w$, $\varphi(w')=su'$, $\varphi(w'')=u''t$. A pair $(u',u'')$ is called \emph{admissible} if it is compatible with at least one interpretation of $u'u''$.

\begin{definition}
	Let $S = (\A,\varphi,W)$ be a DF0L system.
	\begin{enumerate}
		\item A pair $(u',u'')$ is  called \emph{weakly synchronizing} if it is compatible with all interpretations of $u'u''$. We say that the word $u=u'u''$ is \emph{weakly synchronized}.
		\item A pair $(u',u'')$ where $u'\neq\emptyw$ is called \emph{strongly synchronizing} if there is $a\in\A$ such that for all interpretations $(s,w,t)$ of $u'u''$, there is a pair $(w',w'')$ such that $w=w'w''$, $\varphi(w')=su'$, $\varphi(w'')=u''t$, where $w'$ ends with $a$.
	\end{enumerate}
\end{definition}

\begin{remark}\label{r:sync-factor}
	To test whether a pair $(u',u'')$ is weakly or strongly synchronizing, it suffices look at minimal interpretations of $u'u''$.
	Thus those properties are decidable since there is a finite number of minimal interpretation by \Cref{l:inter-length}.
\end{remark}

\begin{lemma}\label{l:extend-sync}
	Let $S$ be a DF0L system and let $(u_1,u_2)$ be such that $u_1u_2\in\lan(S)$. If there is a pair $(v_1,v_2)$ such that $u_1\in \A^*v_1$, $u_2\in v_2\A^*$, and $(v_1,v_2)$ is weakly (strongly) synchronizing, then $(u_1,u_2)$ is also weakly (strongly) synchronizing.
\end{lemma}

\begin{proof}
	Let $S=(\A,\varphi,W)$. Write $u_1=pv_1,u_2=v_2q$ and let $(s,w,t)$ be an interpretation of $u_1u_2$. As $(sp,w,qt)$ is an interpretation of $v_1v_2$ and $(v_1,v_2)$ is synchronizing, there is a pair $(w_1,w_2)$ such that $w_1w_2=w$, $\varphi(w_1) =spv_1=su_1$, $\varphi(w_2) = v_2qt = u_2t$. Thus $(s,w,t)$ is compatible with $(u_1,u_2)$, hence $(u_1,u_2)$ is weakly synchronizing. If  $(v_1,v_2)$ is strongly synchronizing, the last letter of $w_1$ does not depend on $w$, and so $(u_1,u_2)$ is strongly synchronizing as well.
	\qed
\end{proof}

In particular a word $u\in\lan(S)$ which has a weakly synchronized factor is also weakly synchronized. For example, in the Thue--Morse system, any word of $\lan(S)$ which admits $aa$ as a factor is weakly synchronized (see \cref{ex:thue-morse-inter}).

Next we define two notions of circularity. The first one is due to Cassaigne~\cite{Ca94}, while the second one is a variation found in \cite{Klouda2019} of a notion introduced by Mignosi and Séébold~\cite{MiSe}. 
\begin{definition}\label{d:circularity}
	Let $S = (\A,\varphi,W)$ be a DF0L system.
	\begin{enumerate}
		\item $S$ is \emph{weakly circular} if there exists an integer $D\geq 0$ such that all words $u\in\lan(S)$ where $|u|>D$ are weakly synchronized.
		\item $S$ is \emph{strongly circular} if there exists an integer $D'\geq 0$ such that all admissible pairs $(u',u'')$ where $|u'|> D'$ and $|u''|>D'$ are strongly synchronizing.
	\end{enumerate}
\end{definition}

The smallest value of $D$ and $D'$ (if they exist) are called the \emph{thresholds} for weak and strong circularity.
When $S$ is clear from context, they are denoted by $D_\weak$ and $D_\strong$.
By \cref{l:extend-sync}, we can check whether a given value $D$ exceeds $D_\weak$ or $D_\strong$ by checking respectively all words of length $D+1$ and all pairs $(u',u'')$ where $|u'|,|u''|>D+1$. 
Since the property of being weakly or strongly synchronizing is decidable (\cref{r:sync-factor}), this can be turned into algorithms for computing $D_\weak$ and $D_\strong$, if they exist. 
For instance: starting from $D=\maxlen\varphi-2$ (an obvious lower bound for $D_\weak$), one can test whether all words of length $D$ are weakly synchronized, and else increment the value of $D$ and start over. 
We can improve on the brute force approach by reusing calculations from earlier steps, thanks to \Cref{l:extend-sync}. 
Pseudocode can be found in \cref{s:pseudocode}, Algorithms~\ref{lst:weak-sync} and \ref{lst:strong-sync}. 

\begin{example}
	In the Thue--Morse system (\cref{ex:thue-morse}), the word $aba$ is not weakly synchronized (\cref{ex:thue-morse-inter}). 
    But all words of length 4 in $\lan(S)$ admit only one minimal interpretation, thus $D_\weak=3$.
	Similarly, $D_\strong=1$.
\end{example}

\begin{example}\label{ex:eventually-injective-circ}
	For the system of \Cref{ex:eventually-injective}, one can check that $D_\weak=D_\strong=3$. We used this fact to calculate the set $\mathrm\Delta_S$ for this system.
\end{example}

\begin{example}\label{ex:non-eventually-injective-circ}
    The system from \Cref{ex:non-eventually-injective} has $D_\weak=3$ and $D_\strong=9$. 
\end{example}

Thus the previous example is a strongly circular system which is not eventually injective. An example which is weakly but not strongly circular is given in~\cite[Example~5]{Klouda2019}. The next proposition extends an observation from \cite{Klouda2019} which can also be found in \cite[Observation 3.6.14, Proposition 3.6.17]{th/Kyriakoglou2019}.
\begin{proposition}\label{p:weak-vs-strong}
	Let $S=(\A,\varphi,W)$ be a DF0L system.
	\begin{enumerate}
		\item If $S$ is strongly circular, then it is weakly circular with $D_\weak\leq 2D_\strong+\maxlen\varphi$.
		\item If $S$ is eventually injective and weakly circular, then it is strongly circular with $D_\strong\leq D_\weak+\delta_S+1$.
	\end{enumerate}
\end{proposition}

\begin{proof}
	Let $u$ be a word such that $|u|>2D_\strong+\maxlen\varphi$. Let $(s,w,t)$ be an interpretation of $u$ and $w=w_1w_2$ where $w_2$ is the shortest suffix of $w$ such that $|\varphi(w_2)|-|t| > D_\strong$. Write $\varphi(w_1)=su_1$, $\varphi(w_2) = u_2t$, and $w_2 = aw_3$ where $u_1,u_2\in\A^*$ and $a\in\A$. By minimality of $w_2$, we have $|\varphi(w_3)|-|t|\leq D_\strong$ and 
	\begin{equation*}
		|u_1|
		= |u| - |\varphi(a)| - |\varphi(w_3)| - |t|
		\geq |u| - \maxlen\varphi - D_\strong
		> D_\strong.
	\end{equation*}
	Thus $(u_1,u_2)$ is strongly synchronizing, and in particular weakly synchronizing.

	Let $(u_1,u_2)$ be an admissible pair such that $|u_1|=|u_2|=D_\weak+\delta_S+2$. Let $(s,w,t)$ be an interpretation of $u_1u_2$ compatible with $(u_1,u_2)$ and $(\bar s,\bar w,\bar t)$ be any other interpretation of $u_1u_2$. Write $u_1 = u_3r$ where $u_3,r\in\A^*$ and $|r|=\delta_S+1$. By assumption, $u_3$ and $u_2$ are weakly synchronized, so we may find factorizations $u_1 = p_1q_1r$, $u_2=p_2q_2$, $w = w_1w_3w_2$, and $\bar w = \bar w_1\bar w_3\bar w_2$ such that
	\begin{gather*}
		sp_1 = \varphi(w_1),\ \bar sp_1 = \varphi(\bar w_1),\quad
		q_2t = \varphi(w_2),\ q_2\bar t = \varphi(\bar w_2),\\
		q_1rp_2 = \varphi(w_3) = \varphi(\bar w_3).
	\end{gather*}
	Since $|q_1rp_2|>\delta_S$, it follows that that $w_3 = \bar w_3$. Moreover since $(s,w,t)$ is compatible with $(u_1,u_2)$, there must be a factorization $w_3 = xy$ with $x\neq \emptyw$ such that $\varphi(w_1x) = su_1$. It follows that $\bar w = \bar w_1 xy\bar w_2$ with $\varphi(\bar w_1x) = \bar s u_1$. This shows that $(u_1,u_2)$ is strongly synchronizing.
	\qed
\end{proof}

\begin{proposition}\label{p:weak-powers}
	Let $S=(\A,\varphi,W)$ be a PDF0L system and $u\in\lan(S)$. 
	If $u$ is weakly circular in $S^k$ where 
    $|u| > \maxlen\varphi\cdot\max\{ |\varphi^{k-2}(x)| : x\in W\}$,
    then $u$ is also weakly circular in $S$.
\end{proposition}

\begin{proof}
	Let $C= \maxlen\varphi\cdot\max\{ |\varphi^{k-2}(x)| : x\in W\}$. Take a word $u\in\lan(S)$ such that $|u|>C$ which is weakly synchronized in $S^k$.
	Thus there is a pair $(u_1,u_2)$ such that $u_1u_2 = u$ and which is weakly synchronizing in $S^k$.
	Let us show that $(u_1,u_2)$ is also weakly synchronizing in $S$.
	For this, fix an interpretation $(s,w,t)$ of $u$.
	By \Cref{l:inter-length} we have $|w|>|\varphi^{k-2}(x)|$ for every $x\in W$, so $w$ is a factor of $\varphi^{k-1}(z)$ for some $z\in\lan(S)$.
	Let $\varphi^{k-1}(z) = pwq$. Note that $\varphi^k(z) = \varphi(p)sut\varphi(q)$, so $(\varphi(p)s,z,t\varphi(q))$ is an interpretation of $u$ in $S^k$.
	Since $(u_1,u_2)$ is synchronizing in $S^k$, there is a pair $(z_1,z_2)$ such that
	\begin{equation*}
		z = z_1z_2,\quad \varphi^k(z_1) = \varphi(p)su_1,\quad \varphi^k(z_2) = u_2t\varphi(q).
	\end{equation*}
	Since $\varphi^{k-1}(z_1)$ is a prefix of $\varphi^{k-1}(z) = pwq$, it follows that $\varphi^{k-1}(z_1)$ and $pw$ are prefix comparable. If $pw$ would be a proper prefix of $\varphi^{k-1}(z_1)$, then $\varphi(pw) = \varphi(p)sut$ would be a proper prefix of $\varphi^k(z_1) = \varphi(p)su_1$, which is a contradiction (where we used the fact that $\varphi$ is non-erasing). Thus we conclude that $\varphi^{k-1}(z_1)$ is a prefix of $pw$. Likewise it is clear that $p$ must be a prefix of $\varphi^{k-1}(z_1)$. Therefore there is a pair $(w_1,w_2)$ such that $w_1w_2=w$ and $\varphi^{k-1}(z_1) = pw_1$, and as a result $\varphi^{k-1}(z_2) = w_2q$. Finally we have
	\begin{gather*}
		\varphi(p)\varphi(w_1) = \varphi(pw_1) = \varphi^k(z_1) = \varphi(p)su_1,
	\end{gather*}
	and therefore $\varphi(w_1) =su_1$. Likewise we conclude that $\varphi(w_2)=u_2t$, which shows that $(u_1,u_2)$ is compatible with $(s,w,t)$.
	\qed
\end{proof}

In particular, if $S^k$ is weakly circular for some $k\geq 1$, then so is $S$.
The next example shows that the converse is false.

\begin{example}\label{e:non-circ-power}
	Let $S = (\{a,b,c\},\varphi,\{a\})$ where $\varphi\from a\mapsto aac,b\mapsto bc, c\mapsto bc$. The system $S$ is weakly circular with weak threshold 1, but the power $S^2 = (\{a,b,c\},\varphi^2,\{a,aac\})$ is not. For instance, the words $(bc)^{4k}$ have interpretations which never synchronize, such as $(\emptyw,(bc)^k,\emptyw)$ and $(bc,(bc)^kb,bc)$.
\end{example}

\section{Repetitiveness}

In this section, we relate strong circularity with repetitiveness of the system.
Let $S=(\A,\varphi,W)$ be a DF0L system.
We say that a word $w\in\A^*$ is \emph{bounded} if $\lim_{k\to\infty}|\varphi^k(w)|<\infty$, and \emph{unbounded} otherwise.

\begin{definition}
	A DF0L system is called \emph{unboundedly repetitive} if there exists an unbounded word $v \in \lan(S)$ such that $v^k \in \lan(S)$ for all $k\in\nn$.
\end{definition}

Recall that a word $w$ is \emph{primitive} if $u^n=w\implies n=1$. For every word $u\neq\emptyw$, there exists a unique primitive word $\rho(u)$, called the \emph{primitive root} of $u$, such that $u=\rho(u)^n$. It is well known that $\rho(x) = \rho(y)$ if and only if $xy=yx$. When studying repetitiveness of DF0L systems, it is natural to consider the set
\[
	\repets(S) = \left\{ v \in \lan(S) \mid v \text{ is primitive and } \forall k\in\nn, v^k \in \lan(S)  \right\}.
\]

The second and third author proved in \cite{Klouda2019} that in an injective D0L system, strong circularity is equivalent to unbounded repetitiveness. In this section, we will prove the following generalization of this result.

\begin{theorem}\label{t:circ-repet}
	Let $S = (\A, \varphi, W)$ be a PDF0L system.
	\begin{enumerate}
		\item If $S$ is not weakly circular, then it is unboundedly repetitive.
		      \label{i:weakly}
		\item If $S$ is unboundedly repetitive, then it is not strongly circular.
		      \label{i:strongly}
	\end{enumerate}
\end{theorem}

Thanks to \Cref{p:weak-vs-strong} it follows that, when the system is eventually injective, both forms of circularity are equivalent to the absence of unbounded elements in $\repets(S)$, which is a decidable property by~\cite{KlSt13}. 

Observe that unbounded repetitiveness is well-behaved with respect to taking powers, since this does not change the language or the set of unbounded words. Thus we obtain the following corollary of \Cref{t:circ-repet}.
\begin{corollary}
	If a PDF0L system $S$ is strongly circular, then all of the powers $S^n$ for $n\geq 1$ are weakly circular.
\end{corollary}

As shown in \Cref{e:non-circ-power} this may fail when the system is not strongly circular. We do not know if strong circularity is stable under taking powers.

The rest of the section will be devoted to the proof of \Cref{t:circ-repet}, which follows a similar structure as the injective case from \cite{Klouda2019}. Several new steps are needed to account for the lack of injectivity, but some arguments have also been streamlined, in particular in the proof of the first part.

\subsection{Preparatory results for part 1}
\label{s:pre-weakly}

This subsection prepares the proof of part \ref{i:weakly} of \Cref{t:circ-repet}. We start with a simple technical lemma which extends \cite[Lemma~8]{Klouda2019}.
\begin{lemma}\label{l:bounded-sync}
	If $S$ is a PDF0L system, then every long enough bounded word in $\lan(S)$ is weakly synchronized.
\end{lemma}

\begin{proof}
	By \cite[Lemma~8]{Klouda2019}, the result holds for PD0L systems. Applying it to $T = (\A, \varphi, \prod_{w \in W} w)$ and noting that $\lan(S)\subseteq\lan(T)$, it follows that a word in $\lan(S)$ weakly synchronized in $T$ is also weakly synchronized in $S$.
	\qed
\end{proof}

Next we recall some results from \cite{Klouda24}. Let $\varphi: \A^*\to\A^*$ be a morphism and denote by $\alphabet(w)$ the set of all letters which occur in $w$. A subalphabet $\B\subseteq\A$ is called $p$-\emph{invariant} if $\B$ contains an unbounded letter and $\B = \bigcup_{a\in\B}\alphabet(\varphi^p(a))$ \cite[Definition~4]{Klouda24}.
For $p$ fixed, we say that $\B$ is \emph{minimal} if it is minimal for inclusion among all $p$-invariant subalphabets.

\begin{lemma}[{\cite[Lemmas 5 and 6]{Klouda24}}]\label{l:integer-p}
	Let $\varphi$ be a morphism on $\A$. There exists $p\geq 1$ such that 
	$\alphabet(\varphi^p(a)) = \alphabet(\varphi^{pk}(a))$ for any $a \in \A$ and positive $k$.
	In particular, if $a$ is unbounded, $\B = \alphabet(\varphi^p(a))$ is $p$-invariant.
	Moreover, if $\B$ is minimal, then $\B=\alphabet(\varphi^p(g))$ for every unbounded letter $g\in\B$.
\end{lemma}

Given a morphism $\varphi$, we say that a sequence of words $(w_i)_{i\in\nn}$ is \emph{pushy} if the set of bounded words which appear as factors in the words $w_i$, $i\in\nn$, is finite, and otherwise we say that $(w_i)_{i\in\nn}$ is \emph{non-pushy}~\cite[Definition 8]{Klouda24}.

\begin{lemma}[{\cite[Corollary 10]{Klouda24}}]\label{l:non-pushy-sequence}
	Let $\varphi$ be a morphism and $(w_i)_{i\in\nn}$ be a non-pushy sequence of words such that $\lim_{i\to\infty}|w_i|\to\infty$. Let $p$ be the integer of \Cref{l:integer-p}. There exists a minimal $p$-invariant subalphabet $\B$ and an unbounded letter $g \in \B$ with the following property:
	for all $n\in\nn$, there exists $k\in\nn$ such that $|\varphi^{pk}(g)|>n$ and $\varphi^{pk}(g)$ occurs as a factor in some $w_i$, $i\in\nn$.
\end{lemma}

\subsection{Preparatory results for part 2}

We now move on to results used in the proof of part \ref{i:strongly} of \Cref{t:circ-repet}. We first recall the following notion due to Ehrenfeucht and Rozenberg~\cite{Ehrenfeucht1978a}.
\begin{definition}\label{de:simplification}
	Let $\A$ and $\B$ be finite alphabets and $\varphi\from \A^* \to \A^*$, $\psi\from \B^* \to \B^*$, $\alpha\from\A^*\to\B^*$, and $\beta\from\B^*\to\A^*$ be morphisms. We say that $\varphi$ and $\psi$ are \emph{twined} with respect to $(\alpha,\beta)$ if $\beta\circ\alpha = \varphi$ and $\alpha\circ\beta=\psi$. If $\card\B<\card\A$ then we say that $\psi$ is a \emph{simplification} of $\varphi$ with respect to $(\alpha,\beta)$.
\end{definition}

Every non-injective morphism has an injective simplification~\cite[Corollary~1]{Ehrenfeucht1978a}. Moreover if $\varphi$ and $\psi$ are twined with respect to $(\alpha,\beta)$, then for all $k \in \nn$
\begin{equation} \label{eq:simplification_commutes}
	\alpha \circ \varphi^k = \psi^k \circ \alpha, \qquad \varphi^k \circ \beta = \beta \circ \psi^k.
\end{equation}

\begin{lemma}\label{l:simplification-language}
	Let $S = (\A,\varphi,W)$ be a PDF0L system and let $\psi\from\B^*\to \B^*$ be a morphism. Assume that $\varphi$ and $\psi$ are twined with respect to $(\alpha,\beta)$. The system $T = (\B,\psi,\alpha(W))$ satisfies $\alpha(\lan(S))\subseteq\lan(T)$ and $\beta(\lan(T))\subseteq\lan(S)$.
\end{lemma}

\begin{proof}
	First let $v \in \lan(S)$. Then $v$ is a factor of $\varphi^n(w)$ for some $n\in\nn$ and $w \in W$, hence
	$\alpha(v)$ is a factor of $\alpha \circ \varphi^n(w) = \psi^n \circ \alpha(w)$, where the last equality follows from \eqref{eq:simplification_commutes}. By definition of $T$ we have $\alpha(v) \in \lan(T)$. Next let $S' = (\A,\varphi,\varphi(W))$. Then $\lan(S')\subseteq\lan(S)$ and by the first inclusion $\beta(\lan(T))\subseteq\lan(S')\subseteq\lan(S)$.
	\qed
\end{proof}

We say that two words $u$, $v$ are conjugate, denoted $u \sim v$, if $u=xy$ and $v=yx$ for some $x,y\in\A^*$. Equivalently, $u\sim v$ when $|u|=|v|$ and $\rho(u)\sim\rho(v)$. Notice that if $u \sim v$ and $u \in \repets(S)$, then $v \in \repets(S)$. Next we give a series of two lemmas. The first is a rephrasing of \cite[Lemma 4]{KlSt13}.

\begin{lemma}\label{le:pigeonhole}
	Let $\psi\from\B^*\to\A^*$ be an injective morphism and
	$z\in\B^*, v\in\A^*$ such that $v$ is primitive and $\psi(z)$ is a factor of $v^k$, $k\geq 1$.
	If $|z| \geq (\ell+1)|v|$ and $\ell \geq 2$, then
	there is a primitive word $u$ such that
	$u^\ell$ is a factor of $z$ and $\rho(\psi(u)) \sim v$.
\end{lemma}

\begin{proof}
	Take $v_1$ such that $v_1 \sim v$ and $\psi(z)$ is a prefix of $v_1^k$.
	Let $v_1 = p_is_i$ with $|p_i| = i$ for $0\leq i <|v|$. For each $j$, let $q_j$ be the prefix of length $j$ of $z$, and let $k(j)$ and $i(j)$ be the integers such that $\psi(q_j) = {v_1}^{k(j)}p_{i(j)}$.

	As $|z| \geq (\ell+1)|v|$, by the pigeonhole principle, there is a strictly increasing sequence $j_1<j_2<\dots<j_s$ such that $s>\ell$ and $i(j_{t}) = i(j_{1})$ for all $1\leq t\leq s$.
	Observe that $k(j_t)<k(j_{t+1})$, otherwise we would have $\psi(q_{j_t}) = \psi(q_{j_{t+1}})$, which would contradict injectivity.
	If we let $v_2 = s_{i(j_1)}p_{i(j_1)}$, then
	\begin{align*}
		\psi(q_{j_{t+1}}) & = v_1^{k(j_{t+1})}p_{i(j_{t+1})} = v_1^{k(j_{t})}v_1^{k(j_{t+1})-k(j_t)}p_{i(j_{t})}    \\
		                  & = v_1^{k(j_{t})}p_{i(j_t)}v_2^{k(j_{t+1})-k(j_t)} = \psi(q_{j_t})v_2^{k(j_{t})-k(j_1)}.
	\end{align*}

	Let $r_t$ be the suffix of length $j_{t+1}-j_t$ of $q_{j_{t+1}}$, for $t \in \{1,\dots, s-1\}$.
	Then by the above equality we have $\psi(r_t) = v_2^{k(j_{t+1})-k(j_{t})}$. Let $u=\rho(r_1)$. Note that $\rho(\psi(r_t))=\rho(\psi(r_1))=v_2$, which implies that $\rho(r_t)=\rho(r_1)=u$ since $\psi$ is injective. Thus $z$ has a factor $r_1\cdots r_{s-1} = u^n$ for some $n\geq\ell$. Since $\rho(\psi(u)) = \rho(\psi(r_1)) = v_2\sim v$, this concludes the proof.
	\qed
\end{proof}

\begin{lemma}\label{l:lift-repet}
	Let $S = (\A,\varphi,W)$ be a PDF0L system. For all $v\in\repets(S)$, there exists $u\in\repets(S)$ such that $\rho(\varphi(u))\sim v$.
\end{lemma}

\begin{proof}
	First we assume that $\varphi$ is injective. Let $v\in\repets(S)$. For each $\ell\in\nn$, fix a word $z_\ell\in\lan(S)$ such that $|z_\ell|\geq (\ell+1)|v|$ and $\varphi(z_\ell)$ is a factor of some $v^{k}$. By \cref{le:pigeonhole} there exists a primitive word $u_\ell$ such that $u_\ell^\ell$ is a factor of $z_\ell$ and $\rho(\varphi(u_\ell))\sim v$. Since $\varphi$ is injective, it follows that the words $u_\ell$, $\ell\in\nn$ are all conjugate, and so there are only finitely many of them. By the pigeonhole principle, there is $u\in\lan(S)$ such that $u = u_\ell$ for infinitely many $\ell$. In particular, $u\in\repets(S)$ and $\rho(\varphi(u))\sim v$.

	Now let us prove the general case. Let $\psi$ be an injective simplification of $\varphi$ with respect to $(\alpha,\beta)$ and consider the system $T = (\B,\psi,\alpha(W))$. Observe that $\beta$ must be injective since $\psi=\alpha\circ\beta$ is, and likewise $\alpha$ must be non-erasing since $\varphi=\beta\circ\alpha$ is. Let $v \in\repets(S)$.

	First, we claim that there exists $\bar v\in\repets(T)$ such that $\rho(\beta(\bar v))\sim v$. Fix $\ell\in\nn$ and choose $z_\ell\in\lan(S)$ such that $\varphi(z_\ell)$ is a factor of $v^{k}$ for some $k$ and $|z_\ell|\geq (\ell+1)|v|$. Let $\bar z_\ell = \alpha(z_\ell)$. As $\alpha$ is non-erasing we have $|\bar z_\ell|\geq(\ell+1)|v|$. Since $\beta$ is injective,  by \cref{le:pigeonhole}, there is a primitive word $\bar v_\ell$ such that $\bar v_\ell^\ell$ is a factor of $\bar z_\ell$ and $\rho(\beta(\bar v_\ell))\sim v$. As in the first part of the proof, we observe that elements in the set $\{\bar v_\ell \mid \ell\in\nn\}$ are all conjugate, and thus by the pigeonhole principle there exists $\bar v$ such that $\bar v = \bar v_\ell$ for infinitely many $\ell$. Since $\bar z_\ell\in\lan(T)$ by \cref{l:simplification-language}, this concludes the proof of the claim.

	Using the first part of the proof, there exists $\bar u\in\repets(T)$ such that $\rho(\psi(\bar u))\sim\bar v$, and hence, $\rho(\beta\circ \psi(\bar u)) \sim \rho(\beta(\bar v)) \sim v$. Finally if we let $u = \rho(\beta(\bar u))$, then
	\begin{equation*}
		\rho(\varphi(u)) = \rho(\varphi\circ\beta(\bar u)) = \rho(\beta\circ \psi(\bar u)) \sim v
	\end{equation*}
	where we used that $\rho(\varphi(x)) = \rho(\varphi(\rho(x))$. Since $\beta(\bar u)\in \lan(S)$ by \cref{l:simplification-language}, this concludes the proof.
	\qed
\end{proof}

Finally we extend \cite[Theorem~15]{KlSt13} to non-injective systems. Given $u\in\A^*$, let $u^\omega$ denote the infinite periodic word $uuu\cdots \in \A^\nn$.

\begin{proposition}\label{p:unbounded_letter_implies_fp}
	Let $S = (\A, \varphi, W)$ be a PDF0L system.
	Let $v \in \repets(S)$ be unbounded.
	There exists an unbounded letter $a\in\A$, an integer $\ell$, $1\leq\ell\leq\card\A$, and a word $u\in a\A^*\cap\lan(S)$ such that $v \sim u$ and $\lim_{k\to\infty} \varphi^{k\ell}(a) = u^\omega$.
\end{proposition}

\begin{proof}
	Take $j>\card \A(|v|+2)$ such that $|\varphi^j(a)| \geq |v|$ for all unbounded letters $a\in\A$.
	Applying \Cref{l:lift-repet} $j$ times, there exists $z \in \repets(S)$ such that $\rho(\varphi^{j}(z)) \sim v$.
	As $v$ contains an unbounded letter, so does $z$.
	Set $z = p_0a_0w_0$ where $p_0$ is a bounded word, $a_0$ is an unbounded letter and $w_0 \in \A^*$.
	For $i\geq 1$ set $\varphi^i(a_0) = p_ia_iw_i$ where $p_i$ is a bounded word, $a_i$ is an unbounded letter and $w_i\in\A^*$.
	As $\A$ is finite, there exist integers $n$ and $\ell$ with $n<\card\A$ and $\ell\leq\card\A$ such that $a_n = a_{n+\ell}$.
	It follows that $a_i = a_{i+\ell}$ for all $i \geq n$.

	Assume $p_i \neq \varepsilon$ for some $i$ with $n < i \leq n+\ell$.
	Since $j>\card\A(|v|+2)$, $\varphi$ is non-erasing, and $\ell\leq\card\A$, it follows that $\varphi^{j}(z)$ contains a bounded prefix of length greater than $|v|$.
	Since $|\rho(\varphi^{j}(z))| = |v|$ it follows that $\varphi^{j}(z)$ consists only of bounded letters, a contradiction.
	Hence, $p_i = \varepsilon$ for all $i$ with $n < i \leq n+\ell$.
	We conclude that $p_i = \varepsilon$ for all $i$ with $n < i$.
	Finally, since $\varphi^\ell(a_{j-\ell}) = a_{j-\ell}w$ for some $w \in \A^*$ and we can select $j$ arbitrarily large,
	we conclude that $\rho (\varphi^{\ell k}(a_{j-\ell})) \sim v$ for any $k$ such that $|\varphi^{\ell k}(a_{j-\ell})| \geq |v|$.
	\qed
\end{proof}

\begin{lemma} \label{le:fp_w_to_w}
	Let $\varphi\from\A^*\to\A^*$ be a morphism, $a\in\A$, and $w\in\A^*$ such that $\lim_{k\to\infty} \varphi^k(a) = w^\omega$ and $w$ is primitive. Then $\varphi(w) = w^n$ for some $n>1$.
\end{lemma}

\begin{proof}
	Let $p$ be a proper prefix of $w$ and $n\in\nn$ such that
	$\varphi(w) = w^n p$. Since $\lim_{k\to\infty}\varphi^k(a) = w^\omega$, we have $|\varphi(w)| > |w|$ and thus $n\geq 1$. If follows that $\varphi(ww) = w^n p w^n p$ is a prefix of $w^{2n+2}$. By \cite[Property~2.3]{Mosse1992}, this implies $p = \varepsilon$. Thus $\varphi(w) = w^n$, where $n>1$ since $|\varphi(w)|>|w|$.
	\qed
\end{proof}

\subsection{Proof of \Cref{t:circ-repet}}

\paragraph{Part \ref{i:weakly}.}

Assume that $S$ is not weakly circular.
It follows that there exists a sequence $(w_i)_{i\in\nn}$ such that $w_i \in \lan(S)$, $w_i$ is not weakly synchronized, and $|w_i| \to +\infty$ as $i\to\infty$. Since the words $w_i$ cannot have weakly synchronized factors (\Cref{l:extend-sync}), it follows from \Cref{l:bounded-sync} that $(w_i)_{i\in\nn}$ is a non-pushy sequence. Let $p$ be the integer from \Cref{l:integer-p} and let $\psi = \varphi^p$. By \Cref{l:non-pushy-sequence}, we may fix a minimal invariant subalphabet $\B$ and an unbounded letter $g \in \B$ such that for every $n\in\nn$, there is $k\in\nn$ such that $|\psi^{k}(g)|>n$ and $\psi^k(g)$ occurs as a factor in $(w_i)_{i\in\nn}$.
As $\psi^{k}(g)$ is a factor of $w_j$ for some $j$ and $k$, it cannot be weakly synchronized in $S$. By \Cref{p:weak-powers}, it follows that for $n$ large enough, $\psi^k(g)$ is not weakly synchronized in $S^p$.
Moreover, as $\B$ is a minimal invariant subalphabet, it follows that $\psi^{\ell}(g)$ is a factor of $\psi^{k}(g)$ for all $\ell$ with $\ell \leq k$, thus $\psi^\ell(g)$ is not weakly synchronized in $S^p$ for all $\ell\leq k$.
Since this is true for infinitely many integers $k$, we conclude that in fact $\psi^{\ell}(g)$ is not weakly synchronized in $S^p$ for all $\ell\in\nn$.

Fix integers $t,\ell\in\nn$ and put $u = \psi^{t}(g)$.
Observe that $(\varepsilon, \psi^{\ell-1}(u), \varepsilon)$ is an interpretation of $\psi^{\ell}(u)$ in $S^p$ which is compatible with $\left( \varepsilon, \psi^{\ell}(u) \right)$. 
As $\psi^{\ell}(u)$ is not weakly synchronized, we can find an interpretation which is not compatible with the pair $(\varepsilon, \psi^{\ell}(u))$. 
Moreover, observe that $\B\cap\lan(S)$ is extendable (every element has left and right extensions in the language) and thus may take this interpretation to be as long as we need. In particular, we can assume it has the form $(s_\ell, \psi^{\ell-1}(v_\ell), t_\ell)$, for some $v_\ell \in \lan(S)$ such that $s_\ell$ is a prefix of $\psi^{\ell}(a)$ where $a$ is the first letter of $v_\ell$, and $t_\ell$ is a suffix of $\psi^{\ell}(b)$ for the last letter $b$ of $v_\ell$.
Since the interpretation is not compatible with $(\varepsilon, \psi^{\ell}(u))$, we have $s_\ell \neq \varepsilon$.

We claim that the words $v_\ell$ are bounded in length.
Write $v_\ell = a\bar v_\ell b$ where $a,b\in\A$, $\bar v_\ell\in\A^*$. 
Since every letter in $\alphabet(\bar v_\ell)$ is mapped by $\psi^\ell$ to a factor of $\psi(u_\ell)\in\B^*$, \Cref{l:integer-p} implies that $\bar v_\ell\in\B^*$.
Moreover, the sequence $(\psi^{\ell}(u))_{\ell\in\nn}$ is non-pushy, and $\psi^\ell$ maps bounded factors of $\bar v_\ell$ to bounded factors of $\psi^{\ell}(u)$, thus $(\bar v_\ell)_{\ell\in\nn}$ is also non-pushy.
Given $x\in\A^*$, let $\#_\unbound x$ be the number of unbounded letters in $x$. We are reduced to show that $\#_\unbound\bar v_\ell$ is bounded independently of $\ell$.
Since $\B$ is minimal, \Cref{l:integer-p} implies that for all unbounded letters $c\in\B$:
\begin{equation*}
	|\psi^{\ell}(c)|/\lceil\psi \rceil\leq|\psi^{\ell-1}(g)|\leq |\psi^{\ell}(c)|\cdot\lceil\psi \rceil.
\end{equation*}
and therefore
\begin{align*}
	\#_\unbound \bar v_\ell \cdot |\psi^{\ell-1}(g)| \leq |\psi^{\ell}(\bar v_\ell)|\cdot\lceil\psi \rceil \leq |\psi^{\ell}(u)|\cdot\lceil\psi \rceil
	\leq \#_\unbound u\cdot|\psi^{\ell-1}(g)|\cdot\lceil\psi \rceil^2.
\end{align*}
Thus $\#_\unbound v_\ell$ is bounded independently of $\ell$, concluding the proof of the claim.

By the pigeonhole principle, there exists $m$ and $n$ such that $n < m$ and $v_m = v_n = v$,
which yields
\begin{gather}
	\label{eq:m-n}
	\psi^{m-n} (s_n) \psi^{m}(u) \psi^{m-n}(t_n) = \psi^{m}(v) = s_m \psi^{m}(u) t_m, \\
	\label{eq:m-n-1}
	\psi^{m-1}(v) = \psi^{m-n-1}(s_n)\psi^{m-1}(u)\psi^{m-n-1}(t_n).
\end{gather}
Since, by assumption, $(\emptyw,\psi^m(u))$ is not compatible with $(s_m,\psi^{m-1}(v),t_m)$, it follows from \eqref{eq:m-n-1} that $\psi^{m-n}(s_n)\neq s_m$. But by \eqref{eq:m-n}, $\psi^{m-n}(s_n)$ and $s_m$ are prefix comparable, so there is
$z\in\A^+$ such that $s_m = \psi^{m-n}(s_n) z$ or $\psi^{m-n}(s_n)=s_mz$.
In both cases we have from \eqref{eq:m-n} that $\psi^{m}(u)$ is prefix of $z \psi^{m}(u)$. As $\psi^{m-n}(s_n)$ and $s_m$ are prefixes of $\psi^{\ell}(a)$ where $a$ is the first letter of $v$,
we have $|z|<|\psi^{m}(a)|$.

Let us write $u = u_t$, $z = z_t$, $a = a_t$ and $m = m_t$, as all of these objects depend on $t$. As $\B$ is $p$-invariant and $(\psi^{k}(g))_{k\in\nn}$ is a non-pushy sequence, $\psi(g)$ contains at least two unbounded letters.
By \Cref{l:integer-p}, the letter $a_t$ occurs at least $t-1$ times in $\psi^{t}(g)$.
Therefore,
\[
	|\psi^{m_t}(u_t)| = |\psi^{m_t}(\psi^{t}(g))| \geq (t-1) |\psi^{m_t}(a_t)| = (t-1)|z_t|.
\]
Since $z_t\psi^{m_t}(u_t)$ is a prefix of $\psi^{m_t}(u_t)$, this implies that $z_t^{t-1}$ is a prefix of $\psi^{m_t}(u_t)$ for all $t$, and thus $z_t^{t-1}\in \lan(S)$.
As the sequence $(\psi^{\ell}(u_t))_{t\in\nn}$ is non-pushy, $z_t^{t-1}$ must be unbounded, and therefore $S$ is unboundedly repetitive.

\paragraph{Part \ref{i:strongly}.}

Assume that $S$ is unboundedly repetitive, and fix an unbounded element $v\in\repets(S)$. By \Cref{p:unbounded_letter_implies_fp}, there is an unbounded letter $a \in \A$, a positive integer $\ell\leq\card\A$, and a word $u\in a\A^*\cap\lan(S)$ such that $v \sim u$ and $\lim_{k\to\infty} \varphi^{\ell k}(a) = u^\omega$. By \cref{le:fp_w_to_w}, we must have $\varphi^\ell(u) = u^n$ for some $n>1$. Let $j = \min\left\{ i\in\nn \mid \varphi^i(u)\text{ is not primitive} \right \}$. Note that $1\leq j\leq \ell$.
Let $z = \varphi^{j-1}(u)$. Write $\varphi(z) = x^m$ where $x=\rho(\varphi(z))$. It follows from the definition of $j$ that $z\in\repets(S)$ and $m>1$.

By \Cref{p:weak-vs-strong}, if $S$ is not weakly circular, it is not strongly circular.
Thus, we assume that $S$ is weakly circular. For all $i\in\nn$ large enough, the word $x^{im} = \varphi(z^i)$ must be weakly synchronized. It follows that there are words $s,t\in\A^*$ such that $x=st$ and $(s(ts)^{j_1},(ts)^{j_2}t)$ is a weakly synchronizing pair with $j_1+j_2 = im-1$. Let $\bar x = ts$. It follows from \Cref{l:extend-sync} that $(s\bar x^{j_1+k_1},\bar x^{j_1+k_2}t)$ is also weakly synchronizing for all $k_1,k_2\in\nn$. Thus, we may assume that $j_1,j_2$ are both arbitrarily large and, in particular, $j_1 \geq m$.

Since $(\emptyw,z^i,\emptyw)$ and $(x,z^{i+1},x^{m-1})$ are two interpretations of $x^{im}$, there are pairs $(z_1,z_2)$ and $(z_3,z_4)$ such that $z = z_1z_2 = z_3z_4$ and 
\begin{align*}
	\varphi\left(z_1(z_2z_1)^{\ell_1} \right) &= \phantom{x}s\bar{x}^{j_1},&\quad \varphi\left((z_2z_1)^{\ell_2}z_1\right ) &= \bar{x}^{j_2}t, \\
	\varphi\left(z_3(z_4z_3)^{\ell_1} \right) &= xs\bar{x}^{j_1},&\quad\varphi\left((z_4z_3)^{\ell_2+1}z_4 \right) &= \bar{x}^{j_2}tx^{m-1}
\end{align*}
for some $\ell_1,\ell_2\in\nn$ such that $\ell_1+\ell_2 = i-1$.
As $|s\bar{x}^{j_1}| = \left | \varphi\left(z_1(z_2z_1)^{\ell_1} \right) \right | = \left | \varphi\left(z_3(z_4z_3)^{\ell_1} \right) \right | - |\bar{x}|$ and $z_1z_2 = z_3z_4$, we have $|z_1|<|z_3|$.
Since $z$ is primitive, it follows that $z_2z_1\neq z_4z_3$.
Let $q$ be the longest common suffix of $z_2z_1$ and $z_4z_3$. 
Then $s\bar{x}^{j_1}\varphi(q)^{-1}$ is well-defined due to $j_1 \geq m$, and $(s\bar{x}^{j_1}\varphi(q)^{-1}, \varphi(q)\bar{x}^{j_2}t)$ is compatible with both interpretations $(\emptyw,z^i,\emptyw)$ and $(x,z^{i+1},x^{m-1})$ with
\[
	\varphi(z_1(z_2z_1)^{\ell_1}q^{-1}) = s\bar{x}^{j_1}\varphi(q)^{-1}, \varphi(z_3(z_4z_3)^{\ell_1}q^{-1}) = xs\bar{x}^{j_1}\varphi(q)^{-1}.
\]
By the definition of $q$, the last letters of $z_1(z_2z_1)^{\ell_1}q^{-1}$ and $z_3(z_4z_3)^{\ell_1}q^{-1}$ are distinct. This shows that $(s\bar{x}^{j_1}\varphi(q)^{-1}, \varphi(q)\bar{x}^{j_2}t)$ is admissible but not strongly synchronizing. Since $j_1,j_2$ can be arbitrarily large, we conclude that $S$ is not strongly circular.
\qed

\clearpage

\bibliographystyle{splncs04}
\bibliography{biblio.bib}

\clearpage

\appendix

\section{Pseudocode}
\label{s:pseudocode}

\begin{lstlisting}[basicstyle=\small,label={lst:inter},caption={Computes minimal interpretations of words of length $n$ in the language of a PDF0L system.}]
inter = dict(u:set() for u in language(n))
m = 2 + (n - 2) / min(len(morphism(a)) for a in alphabet)

for word x in language(m) do:
    y = morphism(x)
    d = list()

    for i in 0,...,m-1 do:
        for k in 0,...,len(morphism(x[i])) do:
            d.append([i,k])
        end for
    end for

    for i in 0,...,len(y)-n do:
        w = word(x[d[i][0]],...,x[d[i+n-1][0]])
        s = morphism(x[d[i][0]]).prefix(d[i][1])
        t = morphism(x[d[i+n-1][0]]).suffix(d[i+n-1][1])
        inter[y[i:i+n]].add((s,w,t))
    end for

return inter
\end{lstlisting}
\begin{lstlisting}[basicstyle=\small,label={lst:weak-sync},caption={Computes the weak circularity threshold in a weakly circular PDF0L system.}]
d = max(len(morphism(a)) for a in alphabet)-1
inter = interpretations(d)
unsynchronized = set()

function cutting_points(s,w,t):
    return list(len(morphism(p))-len(s) for p in prefixes(w))

for u in language(d) do:
    synchronizing = set(0,...,len(u))
    for (s,w,t) in inter[u] do:
        synchronizing = intersection(visited,cutting_points(s,w,t))
    end for

    if synchronizing is empty do:
        unsynchronized.add(u)
    end if
end for

while unsynchronized is not empty do:
    d = d+1
    new_unsynchronized = set()
    inter = interpretations(d)

    for u in language(d) do:
        if u.prefix(d-1) and u.suffix(d-1) in unsynchronized do:
            synchronizing = set(0,...,len(u))

            for (s,w,t) in inter[u] do:
                synchronizing = intersection(synchronizing,cutting_points(s,w,t))
            end for

            if synchronizing is empty do:
                new_unsynchronized.add(u)
            end if
        end if 
    end for

    unsynchronized = new_unsynchronized
end while

return d-1
\end{lstlisting}
\begin{lstlisting}[basicstyle=\small,label={lst:strong-sync},caption={Computes the strong circularity threshold in a strongly circular PDF0L system.}]
d = -1
unsynchronized = set(word())

while unsynchronized is not empty do:

    d = d+1
    inter = interpretations(2d+2)

    todo = list(v for v in language(2d+2) if v==aub with u in unsynchronized)
    unsynchronized = set()

    for v in todo do:
        weakly_synchronized = undefined

        for s, w, t in inter[v] do:
            j = 0
            for i, a in enumerate(w) do:
                j = j+len(morphism(a))
                if j >= len(s) + d + 1 do:
                    break
                end if
            end for

            if weakly_synchronized is undefined:
                weakly_synchronized = (j == len(s)+d+1)
                if weakly_synchronized is true do:
                    visited_letter = w[i]
                end if
            end if

            if weakly_synchronized is true do:
                if j != len(s)+d+1 or visited_letter!=w[i] do:
                    unsynchronized.add(v)
                    break
                end if

            else do:
                if j == len(s)+d+1 do:
                    unsynchronized.add(v)
                    break
                end if
            end if
        end for
    end for
end while

return d
\end{lstlisting}

\end{document}